\documentclass[aps,pra,twocolumn,nopacs]{revtex4}
\usepackage{amsthm}
\usepackage{amsmath}
\usepackage{latexsym}
\usepackage{amsfonts}
\usepackage{amssymb}
\usepackage{color}
\usepackage{bbm,dsfont}
\usepackage{graphicx}

\usepackage{subfigure}
\usepackage{mathrsfs}
\usepackage{mathbbol}
\usepackage{hyperref}
\usepackage{enumerate}
\usepackage{upgreek}

\newtheorem{proposition}{Proposition}
\newtheorem{proposition?}{Proposition?}
\newtheorem{theorem}{Theorem}

\theoremstyle{definition}

\newcommand{\real}{\mathbb R} 
\newcommand{\complex}{\mathbb C} 
\newcommand{\half}{\tfrac{1}{2}} 
\newcommand{\ii}{i} 

\newcommand{\hi}{\mathcal{H}} 
\newcommand{\hik}{\mathcal{K}} 
\newcommand{\lh}{\mathcal{L(H)}} 
\newcommand{\lin}[1]{\mathcal{L}(#1)} 
\newcommand{\lk}{\mathcal{L(K)}} 
\newcommand{\ip}[2]{\left\langle\,#1\,|\,#2\,\right\rangle} 
\newcommand{\kb}[2]{|#1\rangle\langle#2|} 
\newcommand{\no}[1]{\left\|#1\right\|} 
\newcommand{\tr}[1]{\mathrm{tr}\left[#1\right]} 
\newcommand{\ptr}[2]{\mathrm{tr}_{#1}[#2]} 
\newcommand{\id}{\mathbbm{1}} 

\newcommand{\obs}{\mathfrak{O}}
\newcommand{\A}{\mathsf{A}}
\newcommand{\F}{\mathsf{F}}
\newcommand{\T}{\mathsf{T}}
\newcommand{\X}{\mathsf{X}}
\newcommand{\Y}{\mathsf{Y}}
\newcommand{\Z}{\mathsf{Z}}

\newcommand{\chan}{\mathfrak{C}}

\newcommand{\luders}[1]{\mathcal{L}_{#1}}
\newcommand{\dep}{\Gamma} 

\begin{document}

\title[]{Incompatibility of unbiased qubit observables and Pauli channels}

\author{Teiko Heinosaari}
\affiliation{Turku Centre for Quantum Physics, Department of Physics and Astronomy, University of Turku, FI-20014 Turku, Finland}

\author{Daniel Reitzner}
\affiliation{RCQI, Institute of Physics, Slovak Academy of Sciences, D\'ubravsk\'a cesta 9, 845 11 Bratislava, Slovakia}

\author{Tom\'a\v s Ryb\'ar}
\affiliation{RCQI, Institute of Physics, Slovak Academy of Sciences, D\'ubravsk\'a cesta 9, 845 11 Bratislava, Slovakia}

\author{M\'ario Ziman}
\affiliation{RCQI, Institute of Physics, Slovak Academy of Sciences, D\'ubravsk\'a cesta 9, 845 11 Bratislava, Slovakia}

\begin{abstract}
A quantum observable and a channel are considered compatible if they form parts of the same measurement device, otherwise they are incompatible. 
Constrains on compatibility between observables and channels can be quantified via relations highlighting the necessary trade-offs between noise and disturbance within quantum measurements.  
In this paper we shall discuss the general properties of these compatibility relations, and then fully characterize the compatibility conditions for an unbiased qubit observable and a Pauli channel.
The implications of the characterization are demonstrated on some concrete examples.
\end{abstract}


\maketitle

\section{Introduction}

The paradigm of quantum incompatibility stands behind many quantum phenomena and quantum information no-go theorems \cite{HeMiZi16}. 
One of the most paradigmatic manifestations of incompatibility is the \emph{no-information-without-disturbance theorem}. 
It states that a unitary channel, i.e., a channel that does not cause an irreversible disturbance, is compatible only with trivial observables \cite{KrScWe08ieee,Busch09,ChDaPe10}. 
A trivial observable corresponds to a coin tossing measurement and hence does not give any information on the input state. One therefore concludes that if a measurement gives some information, it must cause disturbance. 

The trade-off between information and disturbance is relevant at least in two different scenarios. 
Firstly, suppose an unwanted disturbance is identified in a communication channel and it may have been caused by the actions of an eavesdropper.
In this case, it is relevant to know what sort of information the eavesdropper may have obtained.
This means that we want to know all possible measurements that may have caused the disturbance.
Secondly, we may plan to perform a measurement of an observable $\A$.
The measurement causes a necessary state perturbation, but the form of this perturbation can be partly controlled by choosing the way in which we measure $\A$.
In this case, it is relevant to know all channels that are compatible with $\A$.

There are several studies in the literature where the information-disturbance relation is investigated by first quantifying information and disturbance and then deriving an inequality for those measures.
In this work we follow a structural approach \cite{HeMi13,HeMiRe14} that does not commit to any specific quantifications of information and disturbance.  
The main idea is to determine if a channel and an observable can be parts of the same measurement process or not.
After presenting the general characterization of compatible pairs of observables and channels (Sec. \ref{sec:incomp}), we concentrate on the cases in which the implemented channel is a Pauli channel (Sec. \ref{sec:pauli}).
In particular, we derive a complete criterion when a noisy version of a binary qubit observable is compatible with a given Pauli channel.
Finally, we demonstrate the consequences of these results on concrete examples (Sec. \ref{sec:examples}).

\section{Incompatibility of channels and observables}\label{sec:incomp}

\subsection{Two equivalent definitions of incompatibility}

When considering (in)compatibility of channels and observables, we can start either from the concept of a \emph{measurement model} or from an \emph{instrument} \cite{QTM96,SSQT01,MLQT12}.
The first one explains the physical meaning of compatibility, while the latter is more convenient from the mathematical point of view. 
In the following we recall these two equivalent ways to define (in)compatibility. 

\begin{figure}
\begin{center}
\includegraphics[width=7cm]{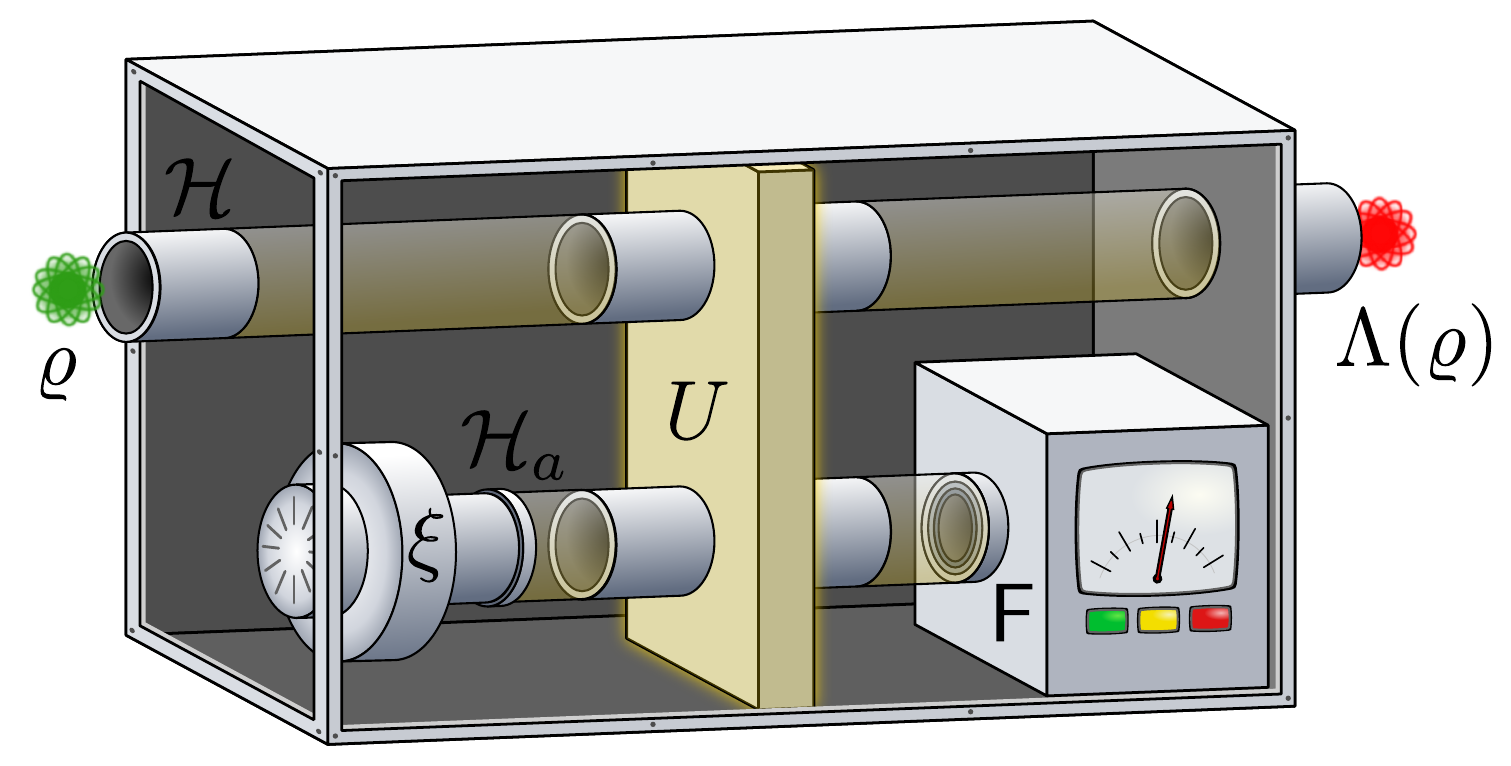}
\end{center}
\caption{\label{fig:memo}A measurement model describes a measurement of a system via coupling it to an ancilliary system. It gives description of both the state change $\Lambda(\varrho)$ and the effective measurement $\A$ on a given input state $\varrho$. This model exemplifies the compatibility of the channel $\Lambda$ and observable $\A$.}
\end{figure}

A measurement model (see Fig.~\ref{fig:memo}) is a formalized description of a measurement process. 
A measured system, associated with a Hilbert space $\hi$, is coupled to an
ancilliary system associated with $\hi_a$.
The composite system undergoes a unitary evolution $\varrho\otimes\xi \mapsto U(\varrho\otimes\xi)U^*$, which is called measurement coupling. 
After that the ancilliary system is measured with a pointer observable
$\F$. 
A measurement model is hence specified by the ancillary Hilbert space $\hi_a$, initial state of the ancilla $\xi$, unitary operator  $U$ on $\hi\otimes\hi_a$ and pointer observable $\F$.
From these it is straightforward to
determine both the average state disturbance $\Lambda$ and
the effective observable $\A$ of the measured system initialized in
the state $\varrho$. 
In particular, for the probability of getting an outcome $x$ we need to have
\begin{equation}\label{eq:memo1}
\tr{\varrho \A(x)} = \tr{ U(\varrho \otimes \xi)U^* (\id \otimes \F(x))} \, .
\end{equation}
This equation is required to be valid for all input states $\varrho$, so it in fact determines the observable $\A$.
Mathematically, $\A$ is a positive operator valued measure (POVM).
The average state disturbance on the system is given as
\begin{equation}\label{eq:memo2}
\Lambda(\varrho) = \ptr{\hi_{a}}{ U(\varrho \otimes \xi)U^*} \, .
\end{equation}
Mathematically, $\Lambda$ is a channel, i.e., a trace preserving completely positive map. 

Physically speaking, the observable $\A$ is related to the information that we can extract from the system's initial state $\varrho$, while the channel $\Lambda$ describes the average perturbation of the state caused by the measurement process, i.e.,~$\Lambda(\varrho)$ is the output state of the system when no post-selection took place.

In what follows we are interested in the inverse problem: given an observable
$\A$ and a channel $\Lambda$, is there a measurement model (specified by a quadruple $(\hi_a,\xi,U,\F)$)
such that $\A$ is given by \eqref{eq:memo1} and $\Lambda$ is given by
\eqref{eq:memo2}? If that is the case, we say that $\A$ and $\Lambda$ are \emph{compatible}; otherwise they are \emph{incompatible}.

The compatibility relation can be concisely expressed in terms of an instrument.
An instrument is a map $x \mapsto \Phi_x$ that assigns a trace-non-increasing completely positive map to each measurement outcome, and such that the sum $\sum_x \Phi_x$ is a trace-preserving map.
The instrument related to a measurement model specified by a quadruple $(\hi_a,\xi,U,\F)$ is
\begin{equation}
\Phi_x(\varrho) =  \ptr{\hi_{a}}{ U(\varrho \otimes \xi)U^*(\id \otimes \F(x))} \, .
\end{equation}
Conversely, for every instrument it is possible to find some measurement model in this way, and one can even choose $\xi$ to be a pure state and $\F$ to be a projection valued measure \cite{Ozawa84}.
The compatibility of a channel $\Lambda$ and an observable $\A$
is hence equivalent to the existence of an instrument $x \mapsto \Phi_x$ such that 
\begin{equation}\label{eq:Kraus}
\sum_x \Phi_x(\varrho) = \Lambda(\varrho) \quad \textrm{and} \quad \tr{\Phi_x(\varrho)}=\tr{\varrho \A(x)}
\end{equation}
for all outcomes $x$ and input states $\varrho$ .

It is reasonable to expect that simultaneous unitary transformations of both the observable and the channel should not change the compatibility relation of the two. This is formalized in the proposition we present after first fixing some notation.
For a unitary operator $V$, we denote by $\widetilde{V}$ the corresponding unitary channel, i.e., 
\begin{equation}
\widetilde{V}(\varrho) = V\varrho V^* \, .
\end{equation}
The functional composition of two channels is denoted by $\circ$.
Hence, for two unitary operators $V,W$ and a channel $\Lambda$, the composition $\widetilde{W}\circ\Lambda\circ \widetilde{V}$ denotes the channel
\begin{equation}
\varrho \mapsto W\Lambda(V\varrho V^*) W^* \, .
\end{equation}
Further, if $\A$ is an observable, then we denote by $V^{\ast} \A V$ the observable consisting of operators $V^{\ast} \A(x) V$.
 
\begin{proposition}\label{prop:unitary}
For any unitary operators $V$ and $W$, the following holds:  
a channel $\Lambda$ is compatible with an observable $\A$ if and only if the
channel $\widetilde{W}\circ\Lambda\circ \widetilde{V}$ is compatible with the observable $V^{\ast} \A V$.
\end{proposition}

\begin{proof}
Suppose that $x\mapsto\Phi_x$ is an instrument such that $\Lambda$ and $\A$ satisfy \eqref{eq:Kraus}.
In that case, the instrument $x\mapsto\Phi'_x$ defined as $\Phi'_x(\varrho) = W\Phi_x(V\varrho V^*) W^*$ demonstrates the compatibility of $\widetilde{W}\circ\Lambda\circ \widetilde{V}$ and $V^{\ast} \A V$.
We can run the same argument for the inverse operators $V^*$ and $W^*$, hence the converse holds also. 
\end{proof}

\subsection{Channels compatible with given observable }\label{sec:least}

Every observable $\A$ has a collection of compatible channels, denoted by $\chan_\A$. The set $\chan_\A$ specifies what kinds of perturbations are possible when $\A$ is measured. There are many ways (by means of measurement models, or instruments) to measure $\A$, and for this reason $\chan_\A$ contains many channels. We will limit $\chan_\A$ to channels that have the same input and output spaces, although generally one could allow arbitrary output spaces \cite{HeMi13}.
For each observable $\A$, the set $\chan_\A$
\begin{itemize}
\item{is \emph{convex}, i.e.,~if $\Lambda_1$ and $\Lambda_2$ are compatible with $\A$, then also all their mixtures $t \Lambda_1 + (1-t) \Lambda_2$, $0<t<1$, are compatible with $\A$,}
\item{is a \emph{left ideal} of the set of all channels (see Fig.~\ref{fig:CA}), i.e.,~if $\Lambda$ is compatible with $\A$ and $\Lambda'$ is any other channel, then their concatenation $\Lambda' \circ \Lambda$ is compatible with $\A$ as well,}
\item{contains all \emph{completely depolarizing channels} (see Fig.~\ref{fig:CA}) $\varrho \mapsto \varrho_0$, where $\varrho_0$ is an arbitrary fixed state}.
 \item{contains the \emph{L\"uders channel} $\luders{\A}$ of $\A$, which is defined as $\luders{\A}(\varrho)=\sum_x \sqrt{\A(x)} \varrho \sqrt{\A(x)}$.}
\end{itemize}

\begin{figure}
\begin{center}
\includegraphics[scale=0.4]{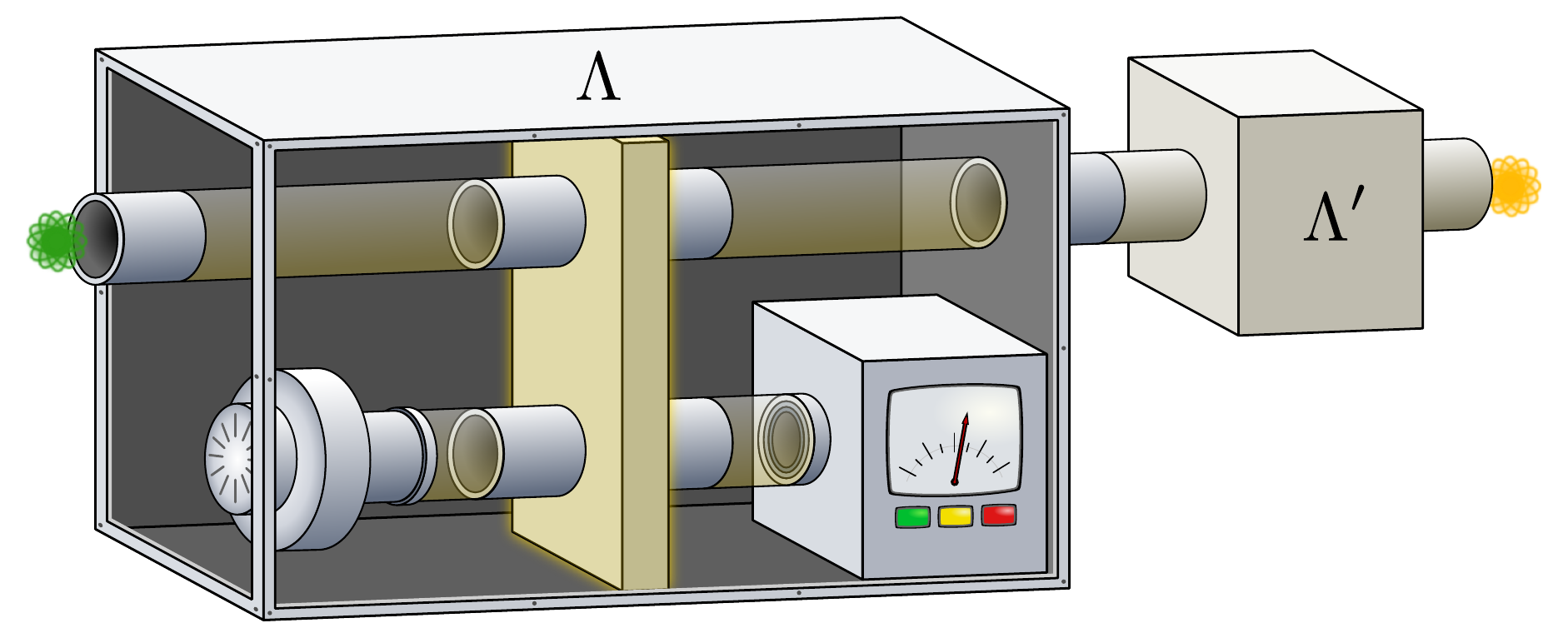}\\\bigskip
\includegraphics[scale=0.4]{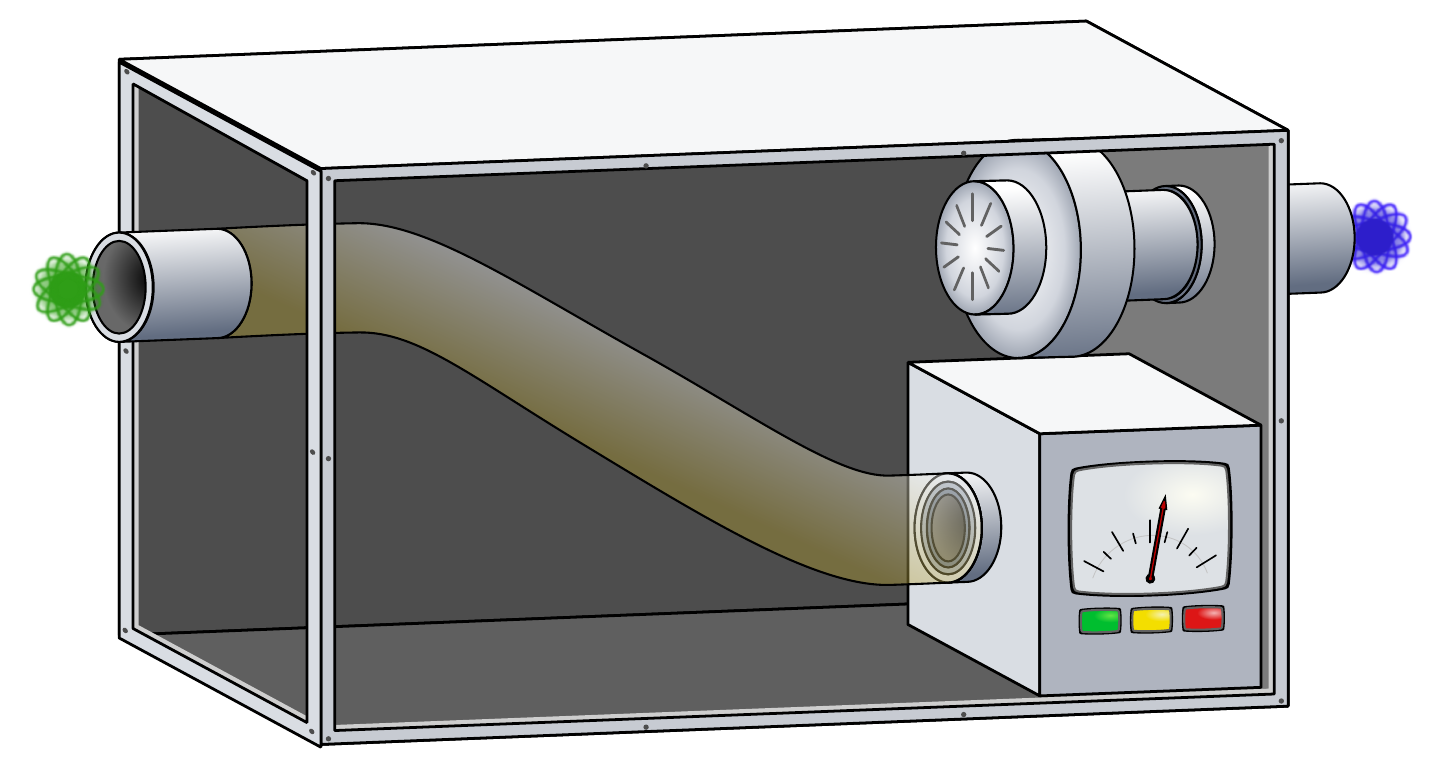}
\end{center}
\caption{\label{fig:CA} The set $\chan_\A$ is a left ideal of the set of all channels --- using some channel $\Lambda'$ after a channel $\Lambda$ compatible with $\A$ produces a new channel compatible with $\A$ (upper figure). The set  $\chan_\A$ always contains completely depolarizing channels $\Gamma$ (lower figure), for which $\Gamma(\varrho)=\varrho_0$ for some fixed state $\varrho_0$. }
\end{figure}

The following result, proved in \cite{HeMi13}, charaterizes the set $\chan_\A$ completely.

\begin{theorem}\label{prop:mother}
  There is a channel $\Lambda_\A:\lin{\hi}\to\lin{\hik}$ and a Hilbert space
  $\hik$  such that every channel $\Lambda$ compatible with $\A$ is of the form
$\Lambda = \Lambda' \circ \Lambda_{\A}$, where $\Lambda':\lin{\hik}\to\lin{\hi}$ is some channel.
\end{theorem}

The 'mother channel' mentioned in Theorem \ref{prop:mother} is given as
\begin{eqnarray}\label{eq:LambdaA} 
\Lambda_{\A} (\varrho) =\sum_x  \hat{\A}(x) T \varrho T^\ast \hat{\A}(x)  \, ,
 \end{eqnarray}
 where $(\hik, \hat{\A},T)$ is some minimal Naimark dilation of $\A$, i.e., $T:\hi\to\hik$ is an isometry, $\hat{\A}$ is a projection-valued measure (PVM) on $\hik$ and $\A(x) = T^* \hat{\A}(x) T 
$ for all $x\in\Omega_\A$.
In a finite dimensional case we can write a concrete form of the minimal Naimark dilation as follows.
We fix a spectral decomposition for each operator $\A(x)$,
\begin{equation}
\A(x)=\sum_{k=1}^{r_x}\kb{\phi_{x,k}}{\phi_{x,k}} \, ,
\end{equation}
where $r_x$ is the rank of $\A(x)$.
We then choose $\hik=\complex^{r_1} \oplus \cdots \oplus \complex^{r_n}$ and we fix orthonormal bases $\{e_{x,k}\}_{k=1}^{r_x}$ for each $\complex^{r_x}$.
We define a linear map $T:\hi\to\hik$ as $T\psi = \sum_{x,k} \ip{\phi_{x,k}}{\psi}e_{x,k}$.
Its adjoint $T^*:\hik\to\hi$ is given as $T^*e_{x,k} = \phi_{x,k}$.
The sharp observable $\hat{\A}$ that dilates $\A$ is 
\begin{equation}
\hat{\A}(x) = \sum_{k=1}^{r_x}\kb{e_{x,k}}{e_{x,k}} \, ,
\end{equation}
hence we obtain
\begin{equation}
\Lambda_\A(\varrho) = \sum_x \sum_{k,\ell}^{r_x} \ip{\phi_{x,k}}{\varrho \phi_{x,\ell}} \kb{e_{x,k}}{e_{x,\ell}} \, .
\end{equation}
By Theorem \ref{prop:mother} any $\Lambda\in\chan_\A$ can thus be written as
\begin{equation}
\Lambda(\varrho) = \sum_x \sum_{k,\ell}^{r_x} \ip{\phi_{x,k}}{\varrho \phi_{x,\ell}} \Lambda'(\kb{e_{x,k}}{e_{x,\ell}}) 
\end{equation}
for some channel $\Lambda':\lin{\hik}\to\lin{\hi}$.

\subsection{Observables compatible with a given channel}\label{sec:rybar}

Let us look at the converse to the previous consideration; we fix a channel $\Lambda$ and denote by $\obs_\Lambda$ the set of all observables compatible with $\Lambda$.
Also the set $\obs_\Lambda$ has some elementary properties, namely, 
\begin{itemize}
\item $\obs_\Lambda$ is \emph{convex}, i.e.,~if $\A_1$ and $\A_2$ are compatible with $\Lambda$, then also all their mixtures $t \A_1 + (1-t) \A_2$, $0<t<1$, are compatible with $\Lambda$, 
\item $\obs_\Lambda$ is \emph{closed under post-processing}, i.e.,~if $\A\in\obs_{\Lambda}$, then also $\mu \circ \A \in\obs_\Lambda$
for all post-processings $\mu$. 
A post-processing $\mu$ is given by formula
\begin{equation}
(\mu \circ \A)(x) = \sum_y \mu_{xy} \A(y) \, , 
\end{equation}
where $\mu_{xy}$ is a stochastic matrix,
\item $\obs_\Lambda$ contains all \emph{trivial observables} $\T$, $\T(x)=p(x)\id$ for some probability distribution $p$.
\end{itemize}

The structure of $\obs_\Lambda$ can be inferred from the results presented in \cite{HeMi17}.
However, we find it useful to give a selfcontained derivation of the charaterization of $\obs_\Lambda$.
To formulate it, we recall that any channel $\Lambda:\lh\to\lh$ has a Stinespring dilation $(\hik,V)$, where $V:\hi \to \hi \otimes \hik$ is an isometry and $\Lambda(\varrho)=\ptr{\hik}{V\varrho V^{\ast}}$.
The dilation also gives another channel $\bar{\Lambda}:\lh\to\lin{\hik}$ by formula $\bar{\Lambda}(\varrho)=\ptr{\hi}{V\varrho V^{\ast}}$.
This channel is called a \emph{conjugate channel} (or complementary channel) of $\Lambda$ (see Fig.~\ref{fig:stinespring}).
We further say that a conjugate channel of $\Lambda$ is minimal if it is
related to a minimal Stinespring dilation of $\Lambda$.

\begin{figure}
\begin{center}
\includegraphics[scale=0.4]{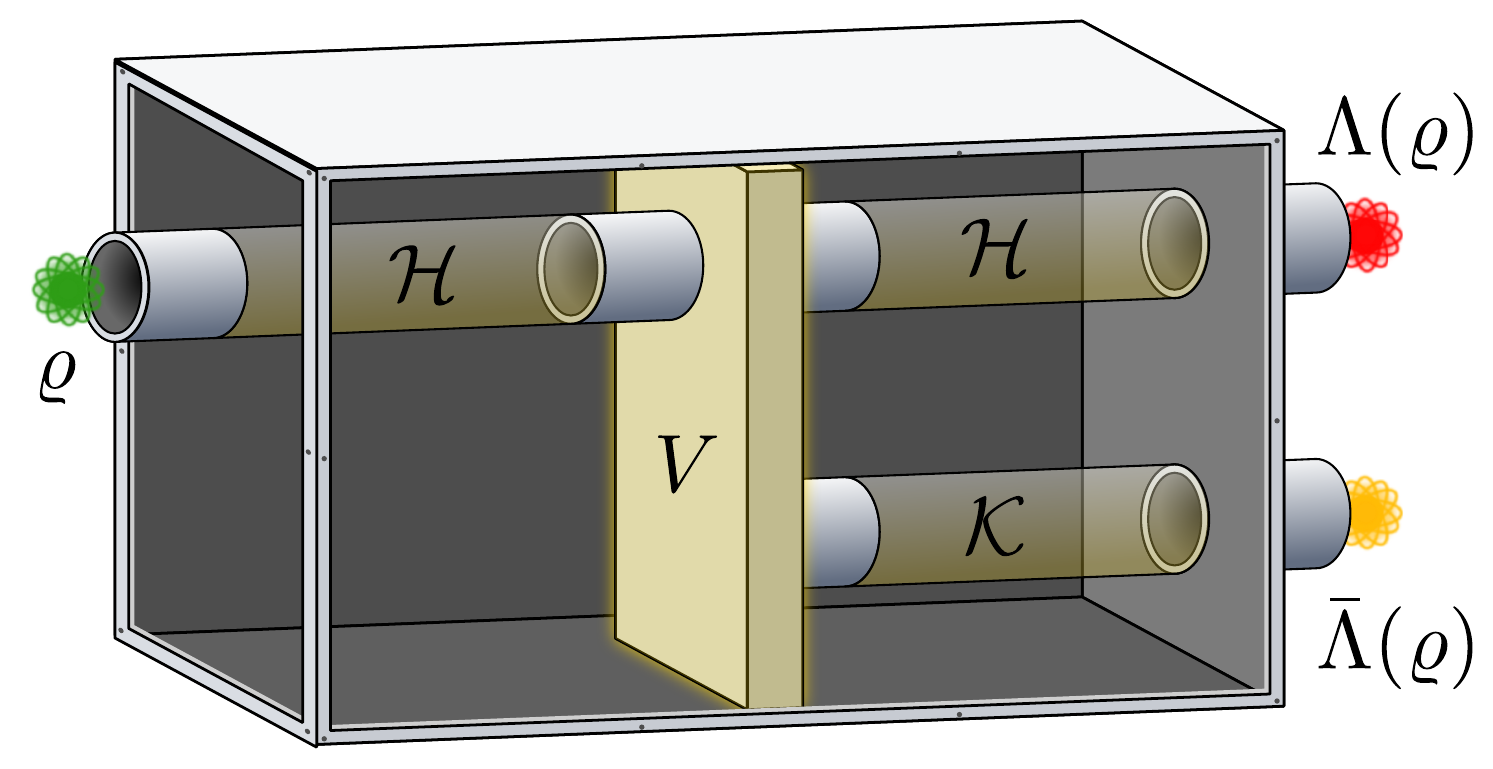}
\end{center}
\caption{\label{fig:stinespring} Every channel $\Lambda$ has a Stinespring dilation expanding the Hilbert space $\hi$ to $\hi\otimes\hik$ by using an isometry $V$. In this way we can define the conjugate channel $\bar{\Lambda}$ of $\Lambda$.}
\end{figure}

\begin{theorem}\label{thm:tuomas}
Let $\Lambda:\lh\to\lh$ be a channel.
Observable $\A$ is compatible with channel $\Lambda$ if and only if the observable can be written in the form $\A(x)=\bar{\Lambda}^{\ast}(\A'(x))$, where $\bar{\Lambda}^{\ast}:\lk\mapsto\lh$ is a fixed (minimal) conjugate channel of $\Lambda$ in the Heisenberg picture and $\A'$ is some observable on $\hik$.
\end{theorem}

\begin{proof}
Assume that $\A$ is compatible with $\Lambda$.
By the definition, this means that the conditions \eqref{eq:memo1} and \eqref{eq:memo2} hold for some $\hi_{a}$, $U$, $\xi$ and $\F$, where $\xi$ is a pure state, i.e., $\xi=\kb{\phi}{\phi}$.
We then define an operator $W:\hi\to\hi\otimes\hi_{a}$ as
\begin{equation}
W\psi = U (\psi \otimes \phi) \, .
\end{equation}
This operator satisfies 
\begin{equation}
W^*W=\id_\hi \quad \textrm{and} \quad W\varrho W^* = U(\varrho \otimes \xi)U^*
\end{equation}
 for all states $\varrho$ on $\hi$.
We conclude that $(\hi_{a},W)$ is a Stinespring dilation of $\Lambda$ and
\begin{equation}
\A(x)=W^{\ast}(\id_{\hi}\otimes\F(x))W \, .
\end{equation}
Let us then fix the minimal Stinespring dilation $(\hik,V)$ for the channel $\Lambda$. 
From the minimality follows that $W=(\id_{\hi}\otimes V')V$, where $V':\hik\to\hi_{a}$ is an isometry.
Therefore,
\begin{equation}\label{eq:prop2obs}
\A(x)=V^{\ast} (\id_{\hi}\otimes V'^*\F(x)V')V \, .
\end{equation}
Denoting $\A'(x) = V'^*\F(x)V$ we obtain the claimed form.

Conversely, assume that $\A'$ is some observable on $\hik$ and $\A(x)=\bar{\Lambda}^{\ast}(\A'(x))$ for some (minimal) Stinespring dilation $(\hik,V)$ of $\Lambda$.
Then we define an instrument
\begin{equation}
\Phi_x(\varrho) =  \ptr{\hik}{ V\varrho V^*(\id_\hi\otimes \A'(x)) } \, .
\end{equation}
This instrument fullfills the conditions in \eqref{eq:Kraus}, hence showing that $\A$ and $\Lambda$ are compatible.
\end{proof}

Let us note that the condition on the minimality of the dilation  in Theorem \ref{thm:tuomas}  is not necessary. However, we found it convenient to use a concrete form of the minimal Stinespring dilation $(\hik,V)$ of $\Lambda$. In particular, fix an orthonormal basis $\{e_k\}_{k=1}^N$ for $\hik$ and for each $k$ we define an operator $M_k\in\lh$ via 
\begin{equation}
\ip{\psi}{M_k\varphi}=\ip{\psi \otimes e_k}{ V \varphi} \, , \quad \psi,\varphi\in\hi \, .
\end{equation}
Then the operators $M_k$ form a minimal set of Kraus operators for the channel $\Lambda$ and
we obtain
\begin{equation}\label{eq:tuomas}
\A(x)=\sum_{k,l} \ip{e_k}{\A'(x) e_l} M_k^{\ast}M_l \, .
\end{equation}
Here again $\A'$ is some observable on $\hik$.

\section{Incompatibility of unbiased qubit observables and Pauli channels}\label{sec:pauli}

\subsection{Unbiased qubit observables and Pauli channels}

An observable $\A$, acting on a Hilbert space $\hi$, is called \emph{unbiased} if it maps the maximally mixed state $\frac{1}{d} \id$ to the uniform probability distribution of outcomes $x$, i.e.,~$\tr{\A(x)}=\frac{d}{n}$, where $n$ is the number
of outcomes and $d$ is the dimension of $\hi$. 
In the case of qubit observables the unbiasedness condition implies
that effects have the form $\A(x)=\frac{1}{n}[\id +
\vec{a}(x)\cdot\vec{\sigma}]$, where $\vec{a}(x)\in\real^3$ and
$\vec{\sigma}=(\sigma_1,\sigma_2,\sigma_3)$
is the vector of Pauli operators. In what follows we will use the notation
$\sigma_0=\id$.

In particular, for binary (i.e. two outcome) qubit observables
the unbiasedness condition means that observables are of the form
\begin{align}
\A_{s,\vec{n}}(\pm) & = \half \bigl( \id \pm s \vec{n}\cdot\vec{\sigma} \bigr)
\end{align}
for $s\in[0,1]$ and $\vec{n}\in\real^3$, $\no{\vec{n}}=1$. 
We will also use the notation $\X_s$, $\Y_s$ and $\Z_s$ for observables $\A_{s,\vec n}$ with $\vec{n}=(1,0,0)$, $\vec{n}=(0,1,0)$ and $\vec{n}=(0,0,1)$, respectively.

We notice that an observable $\A_{t,\vec{n}}$ is a post-processing of another observable $\A_{s,\vec{n}}$ if and only if $t\leq s$. 
Namely, if $s\neq 0$, then for any $t\in [0,1]$ the operator $\A_{t,\vec{n}}(+)$ can be written as a linear combination of $\A_{s,\vec{n}}(+)$ and $\A_{s,\vec{n}}(-)$ in a unique way,
\begin{equation}
\A_{t,\vec{n}}(+) = \frac{s + t}{2s} \A_{s,\vec{n}}(+) + \frac{s - t}{2s} \A_{s,\vec{n}}(-) \, .
\end{equation}
This is a valid post-processing if and only if 
\begin{equation}
0\leq (s \pm t)/2s \leq 1
\end{equation}
which is equivalent to $t \leq s$. 
We can therefore interpret the parameter $s$ as the degree of noise inherent in $\A_{s,\vec{n}}$.

Let us note that the set of effects $\A_{s,\vec{n}}(\pm)$ (thus also the set of unbiased binary qubit observabeles) is convex. Indeed the effects $\A_{s,\vec{n}}(\pm)$ are positive operators of unit trace, hence, they formally correspond to density operators and as such, can be visualized as points inside the Bloch ball.
An unbiased binary qubit observable thus corresponds to a pair of points inside the Bloch ball, and the points are symmetric with respect to the origin. 

Using the analogy with observables, we say that a channel $\Lambda$ is unbiased if it keeps the maximally mixed state invariant, i.e., $\Lambda(\tfrac{1}{d}\id)=\tfrac{1}{d}\id$. 
This property is obviously equivalent with unitality, thus, the notion of unbiased channels is just
a synonym for \emph{unital channels}.
In the case of qubits it is further known that the set
of unital channels coincides with the set of \emph{random unitary channels} \cite{RuSzWe02,GQS06}.
A prominent class of random unitary qubit channels are the so-called Pauli channels, and in the following we shall concentrate on that class. 

A \emph{Pauli channel} $\Psi_{\vec{p}}$ is a qubit channel of the form
\begin{equation}
\Psi_{\vec{p}}(\varrho) = \sum_{j=0}^3 p_j \, \sigma_j \varrho \sigma_j  \, , 
\end{equation}
where $\vec{p}\in\real^4$ is a probability vector, i.e.~$0 \leq p_j \leq 1$ and $\sum_{j=0}^3 p_j = 1$.
Due to the normalization of $\vec{p}$, a Pauli channel $\Psi_{\vec{p}}$ is determined already by three of the components, e.g.,~$p_1,p_2,p_3$.
We can therefore visualize the set of Pauli channels as a tetrahedron
in $\real^3$; see Fig. \ref{fig:triangle}. 
We denote by $\Gamma$ the completely depolarizing channel on the maximally mixed state $\frac{1}{2}\id$, and it corresponds to the probability vector $\vec{p}=(\tfrac{1}{4},\tfrac{1}{4},\tfrac{1}{4},\tfrac{1}{4})$.

\begin{figure}
\begin{center}
\includegraphics[width=5.0cm]{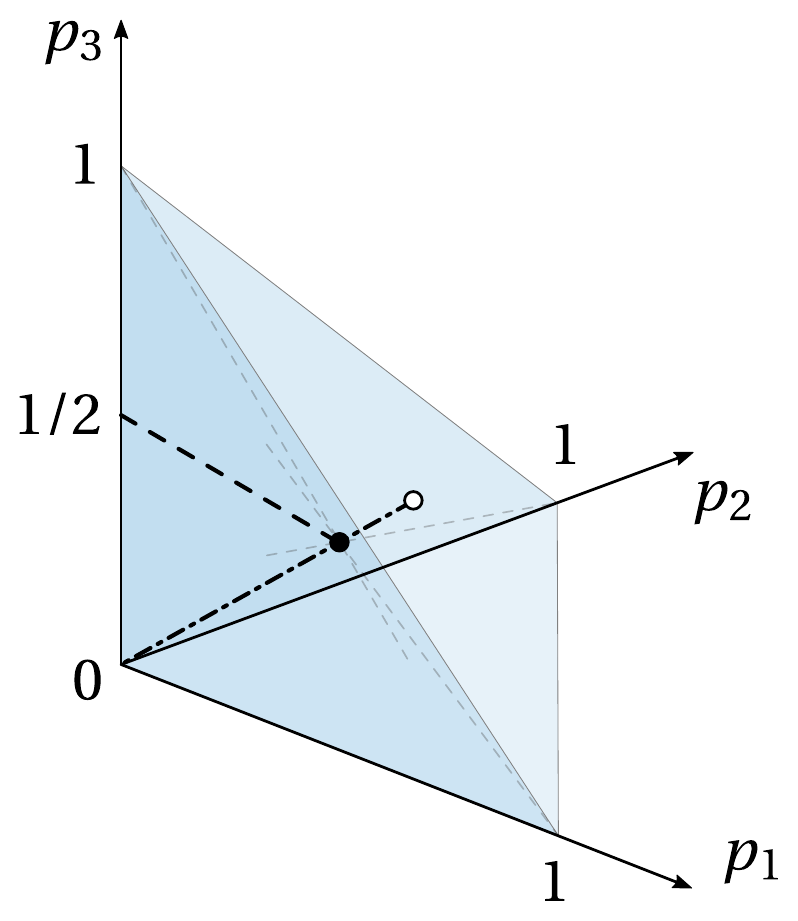}
\end{center}
\caption{\label{fig:triangle} Pauli channels can be parametrized by the points inside the tetrahedron.
The identity channel corresponds to the origin. The three Pauli unitaries correspond to the remaining three vertices of the tetrahedron. 
The edges correspond to Pauli channels having two of the probabilities zero. Special classes of Pauli channels are: (1) the class of partially depolarizing channels (dot-dashed line) including the totally depolarizing channel (solid point) and the quantum NOT (hollow point); (2)  the class of measure and prepare channels (dashed line), here shown for measurements in the $z$-direction.
}
\end{figure}

\subsection{Compatibility condition}
Let $\A$ be a qubit observable compatible with a Pauli
channel $\Psi_{\vec{p}}$.  By concatenating $\Psi_{\vec{p}}$ with a
Pauli unitary channel we generate three new Pauli channels that are
compatible with $\A$. 
Using Prop.~\ref{prop:unitary} with $W=\sigma_i$ and $U=\id$ we conclude the following.

\begin{proposition}\label{prop:permu}
Let $\A$ be a qubit observable compatible with a Pauli channel $\Psi_{\vec{p}}$, $\vec{p}=(p_0,p_1,p_2,p_3)$. Then $\A$ is also compatible with Pauli channels with the following probability vectors:
\begin{itemize}
\item $(p_1,p_0,p_3,p_2)$,
\item $(p_2,p_3,p_0,p_1)$,
\item $(p_3,p_2,p_1,p_0)$.
\end{itemize}
\end{proposition}

In conclusion, for a fixed qubit observable $\A$, the probability vectors $\vec{p}=(p_0,p_1,p_2,p_3)$ that correspond to Pauli channels $\Psi_{\vec{p}}$ compatible with the observable form a convex region inside of the tetrahedron in Fig. \ref{fig:triangle}, and this region has the permutational symmetry described in Prop.~\ref{prop:permu}.

Let then $\Lambda$ be a qubit channel compatible with an unbiased qubit observable $\A_{s,\vec{n}}$. 
As we have seen earlier, any observable $\A_{t,\vec{n}}$ with $t \leq s$ is a post-processing of $\A_{s,\vec{n}}$. 
It follows that also $\A_{t,\vec{n}}$ is compatible with $\Lambda$. 
For any unit vector $\vec{n}$, it thus makes sense to seek for the largest $s$ such that $\A_{s,\vec{n}}$ and $\Lambda$ are compatible. 

The next result gives a sufficient and necessary condition for an unbiased binary qubit observable $\A_{s,\vec{n}}$ and Pauli channel $\Psi_{\vec{p}}$ to be (in)compatible.
The compatibility properties of a Pauli channel $\Psi_{\vec{p}}$ are determined by the vector $\vec{p}$, and to formulate the condition we denote
\begin{eqnarray}
{p}_{\pm}[1] &:=& 2(\sqrt{p_0p_1}\pm\sqrt{p_2p_3}) \, ,\notag\\
{p}_{\pm}[2] &:=& 2(\sqrt{p_0p_2}\pm\sqrt{p_1p_3}) \, ,  \\
{p}_{\pm}[3] &:=& 2(\sqrt{p_0p_3}\pm\sqrt{p_1p_2}) \, .\notag
\end{eqnarray}
We observe that $|p_-[j]|\leq p_+[j]$ for $j=1,2,3$. 
Furthermore all $p_j[+]$ are invariant under the permutations described in Prop.~\ref{prop:permu}.

\begin{theorem}\label{prop:maincomp}
An unbiased binary qubit observable $\A_{s,\vec{n}}$ and a Pauli channel $\Psi_{\vec{p}}$ are compatible if and only if
\begin{align}\label{eq:sufnec}
\frac{s^2n_1^2}{{p}_+[1]^2}+\frac{s^2n_2^2}{{p}_+[2]^2}+\frac{s^2n_3^2}{{p}_+[3]^2}\leq 1 \, .
\end{align}
This inequality is understood in a way that if ${p}_+[j]=0$, then necessarily
the whole term vanishes.
\end{theorem}

Before we present the proof of this statement, let us discuss the content of Ineq.~\eqref{eq:sufnec}.
Firstly, suppose that a Pauli channel $\Psi_{\vec{p}}$ is fixed and that ${p}_+[j]\neq 0$ for every $j=1,2,3$.
For a vector $s\vec{n}$, the inequality \eqref{eq:sufnec} is a solid ellipsoid (see Fig.~\ref{fig:ellipsoid}).
\begin{figure}
\begin{center}
\includegraphics[width=3.5cm]{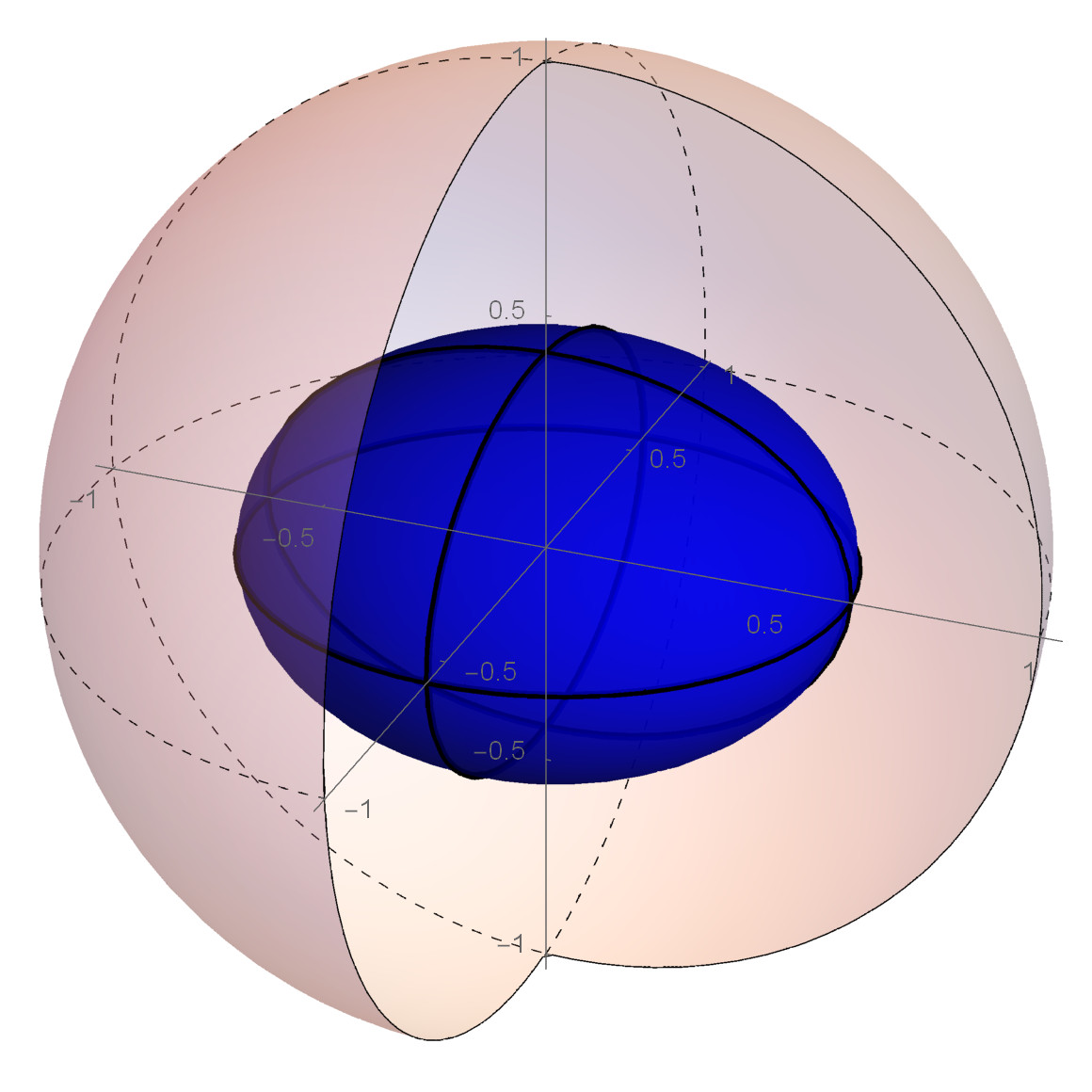}\hskip5mm\includegraphics[width=3.5cm]{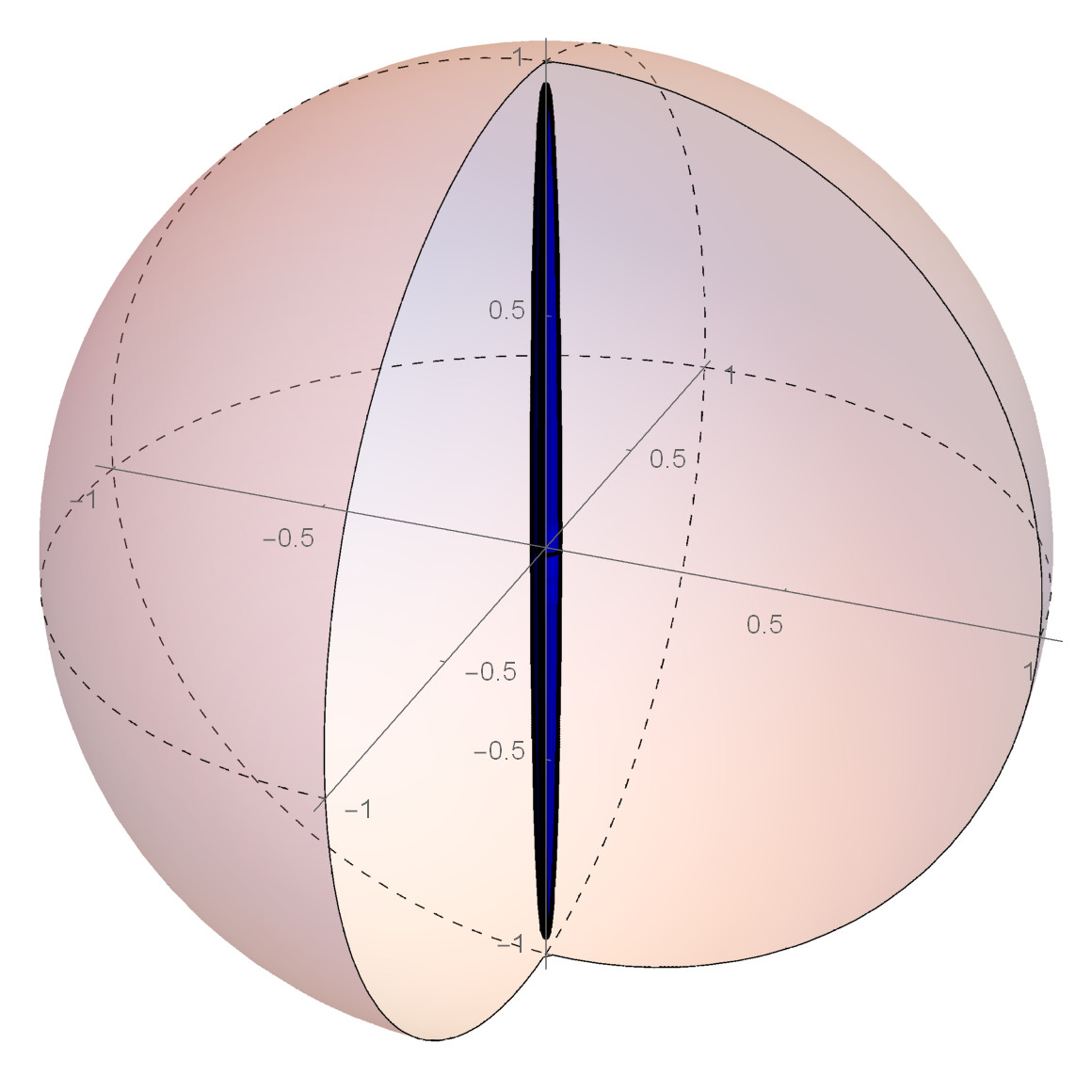}
\end{center}
\caption{\label{fig:ellipsoid}Compatibility region for Pauli channel $\Psi_{\vec{p}}$ and observable $\A_{s,\vec{n}}$ is in general described as an ellipsoid for allowed Bloch vectors of the observable (left figure). In specific cases it collapses to a line (right figure) as e.g.~in the case of phase damping channels.}
\end{figure}
Secondly, for at least one of $p_+[j]$ to be zero, we need to have at least two of the components of $\vec p$ zero, which always makes at least two of $p_+[j]$ zero. In such case the inequality \eqref{eq:sufnec} does not represent a solid ellipsoid but only a line in one of the canonical directions. 
In the most extreme case we have ${p}_+[j]=0$ for all $j=1,2,3$, which occurs when only one of the $p_j$'s is non-zero. Then the ellipsoid collapses to a point $s\vec{n}=\vec{0}$. 
This equation is satisfied only when $s=0$, and this is consistent with our earlier discussion that only a trivial observable is compatible with a unitary channel.

\begin{figure}
\begin{center}
\includegraphics[width=7cm]{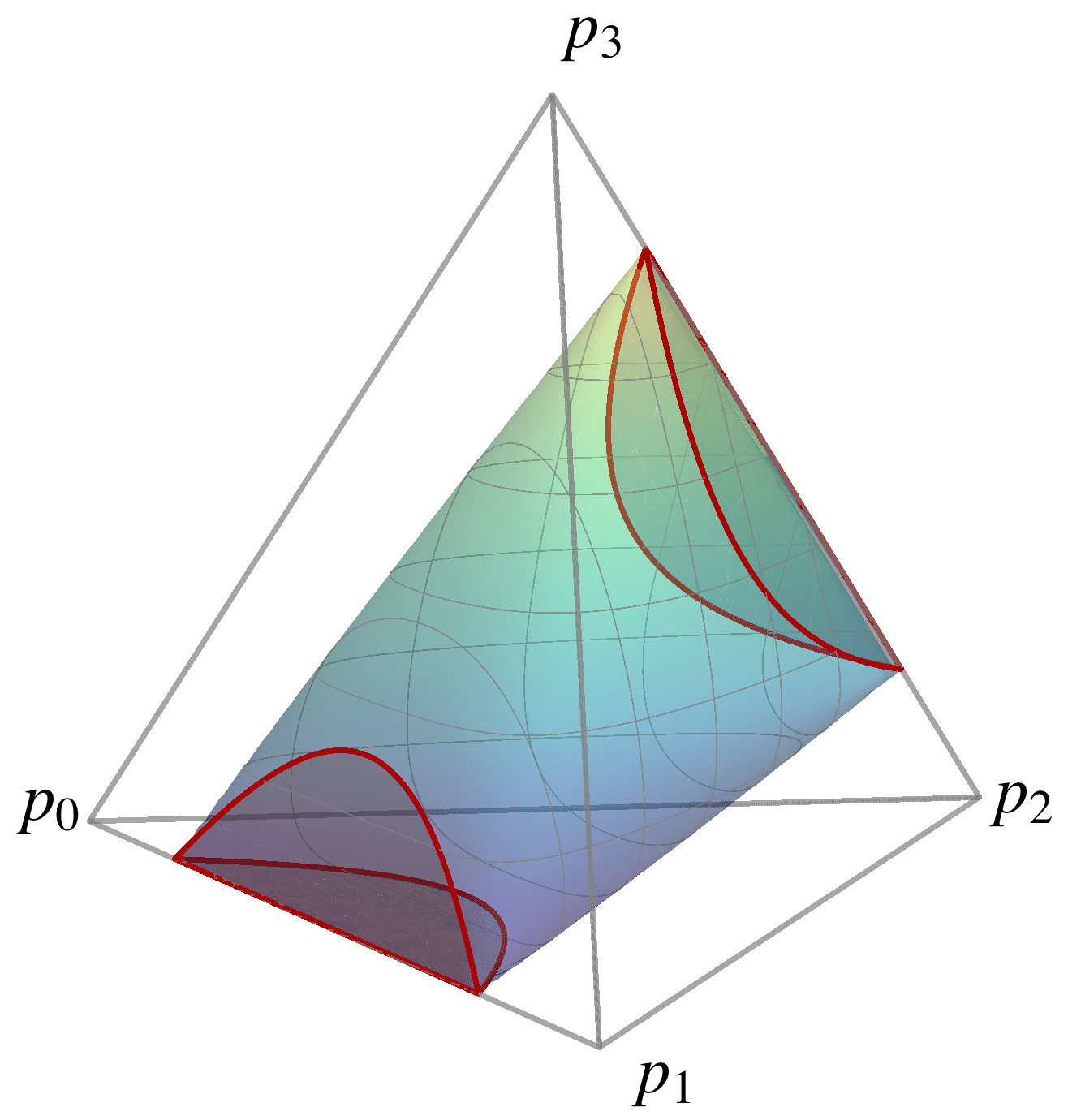}
\end{center}
\caption{\label{fig:uglyregion}Compatibility region for Pauli channel $\Psi_{\vec{p}}$ and observable $\X_{s}$ for allowed $\vec p$ vectors of the Pauli channel. Here, unlike in Fig.~\ref{fig:triangle}, the region of vectors $\vec p$ is depicted as a simplex. The darkened regions are parts laying on the edges of the simplex, i.e.~when (at least) one of the $p_j$'s equals zero.}
\end{figure}

While for constant Pauli channel the set of compatible unbiased qubit observables is rather simply described and visualized, the region of Pauli channels compatible with some given observable that we get from (\ref{eq:sufnec}) is difficult to describe. One such is presented in Fig.~\ref{fig:uglyregion} for observable $\X_s$ for $s=0.8$.

\subsection{Proof of Theorem \ref{prop:maincomp}}

Let $\Psi_{\vec{p}}$ be a Pauli channel.
In the following we will assume that $p_j\neq 0$ for every $j=0,\ldots,3$. 
The required modifications to the proof in the other cases shall be obvious. 

The minimal set of Kraus operators for $\Psi_{\vec{p}}$ is given as $M_k=\sqrt{p_k}\sigma_k$, $k=0,1,2,3$.
Using Theorem \ref{thm:tuomas} and the formula \eqref{eq:tuomas}
derived from it, we obtain observables compatible with $\Psi_{\vec{p}}$ by inserting various choices for $\A'$, which is an observable acting on $\complex^4$.

We fix a unit vector $\vec{n}$ and we seek for allowed $s$ such that $\A_{s,\vec{n}}$ and $\Psi_{\vec{p}}$ are compatible.
It is useful to define a vector $\vec{n}^{\prime}\in\real^3$ as 
\begin{equation}\label{eq:nprime}
n'_j =\frac{n_j}{p_+[j]} \left(\frac{n_1^2}{p_+[1]^2} + \frac{n_2^2}{p_+[2]^2} + \frac{n_3^2}{p_+[3]^2}\right)^{-1/2} \, ,
\end{equation}
then an operator $\A^{\prime}(+)$ as
\begin{align}
\label{eq:suffA}
\A^{\prime}(+)=&\frac{1}{2}\left(
\begin{array}{cccc}
1 & n'_1 & n'_2 & n'_3 \\
n'_1 & 1 & -\ii n'_3 & \ii n'_2\\
n'_2 & \ii n'_3 & 1 & -\ii n'_1\\
n'_3 & -\ii n'_2 & \ii n'_1 & 1 
\end{array}\right) \, , 
\end{align}
and set $\A^{\prime}(-) = \id - \A^{\prime}(+)$. 
Since $\vec{n}^{\prime}$ is a unit vector, one can verify that $\A^{\prime}(\pm)$ are projections and hence form a binary observable. 
Applying equation \eqref{eq:tuomas} we then get
\begin{multline}
\sum_{k,l} \ip{e_k}{\A^{\prime}(x) e_l} \sqrt{p_kp_l}\sigma_k\sigma_l \\
=\frac{1}{2}\left[\id + \left(\frac{n_1^2}{p_+[1]^2} + \frac{n_2^2}{p_+[2]^2} + \frac{n_3^2}{p_+[3]^2}\right)^{-\tfrac{1}{2}}(\vec{n}\cdot\vec{\sigma})\right]
\end{multline}
This means that $\A_{s,\vec{n}}$ and $\Psi_{\vec{p}}$ are compatible for 
\begin{equation}\label{eq:smax}
s= s_{max} \equiv \left(\frac{n_1^2}{p_+[1]^2} + \frac{n_2^2}{p_+[2]^2} + \frac{n_3^2}{p_+[3]^2}\right)^{-\tfrac{1}{2}} \, , 
\end{equation}
and hence also for any value $s\leq s_{max}$.
We have thus seen that the inequality \eqref{eq:sufnec} is a sufficient condition for the compatibility of $\A_{s,\vec{n}}$ and $\Psi_{\vec{p}}$.

In order to prove that the inequality \eqref{eq:sufnec} is also necessary for the compatibility of $\A_{s,\vec{n}}$ and $\Psi_{\vec{p}}$, we will formulate the problem in terms of a semidefinite program (SDP). 
Any feasible instance of primal SDP problem will give us a lower bound on the largest possible $s$, while any feasible instance of the dual SDP problem will give us an upper bound on the largest possible $s$. 
If both bounds coincide, we have found the optimal solution with the largest possible $s$. And this will be indeed the case.

\subsubsection*{Primal SDP problem}
For given vectors $\vec{n}$ and $\vec{p}$, we want to find the largest $s$ such that $\A_{s,\vec{n}}$ and $\Psi_{\vec{p}}$ are compatible. 
Since $\A_{s,\vec{n}}(-)=\id-\A_{s,\vec{n}}(+)$, we can formulate the question in terms of the effect $\A_{s,\vec{n}}(+)$ only. 
Further, we will now understand $\vec{n}$ and $s$ as parameters of the SDP problem and will omit the subscripts from now on. 
We will thus denote $A \equiv \A_{s,\vec{n}}(+)$, and we observe that $s=\tr{A (\vec{n}\cdot\vec{\sigma})}$.

In summary, we are trying to find the maximum of $\tr{A (\vec{n}\cdot\vec{\sigma})}$ over all effects $A$ such that the corresponding binary observable $\A_{s,\vec{n}}$ is compatible with $\Psi_{\vec{p}}$. 
Using Theorem \ref{thm:tuomas}, this is equivalent to  
\begin{align}
\max_{\id\geq \A^{\prime}\geq 0}\tr{\overline{\Psi}^{\ast}_{\vec{p}}(A^{\prime})(\vec{n}\cdot\vec{\sigma})} \, , 
\end{align}
where $A^{\prime}$ is an  effect on the Hilbert space defined by the minimal conjugate channel $\overline{\Psi}^{\ast}_{\vec{p}}$ of $\Psi_{\vec{p}}$. 
We take the conjugate channel that is related to the Kraus operators $M_k=\sqrt{p_k}\sigma_k$ of $\Psi_{\vec{p}}$ and denote $\Sigma_i=\overline{\Psi}_{\vec{p}}(\sigma_i)$. 
Then 
\begin{equation}
\Sigma_i = \sum_{j,k,n=0}^3 \sqrt{p_kp_n} \tr{\sigma_k \sigma_i \sigma_n} \kb{e_k}{e_n} \, , 
\end{equation}
and we obtain
\begin{align}
\Sigma_1=2\left(\begin{array}{cccc}
0 & \sqrt{p_0p_1} & 0 & 0\\
\sqrt{p_0p_1} & 0 & 0 & 0\\
0 & 0 & 0 & -\ii\sqrt{p_2p_3}\\
0 & 0 & \ii\sqrt{p_2p_3} & 0
\end{array}\right),\\
\Sigma_2=2\left(\begin{array}{cccc}
0 & 0 &\sqrt{p_0p_2} & 0 \\
0 & 0 & 0 &\ii\sqrt{p_1p_3}\\
\sqrt{p_0p_2} & 0 & 0 &0\\
0 &  -\ii\sqrt{p_1p_3} & 0 & 0
\end{array}\right),\\
\Sigma_3=2\left(\begin{array}{cccc}
0 & 0 &0 & \sqrt{p_0p_3} \\
0 & 0 &  -\ii\sqrt{p_1p_2} &0\\
0 & \ii\sqrt{p_1p_2} & 0 &0\\
\sqrt{p_0p_3} & 0 & 0 & 0
\end{array}\right).
\end{align}

We can finally write the primal SDP problem as
\begin{align}
s_{\textrm{P}}:=&\max_{A'} \tr{A^{\prime} (\vec{n}\cdot\vec{\Sigma})}\label{eq:primal}\\
&
\begin{aligned}\notag
\text{subject to}\qquad & -A^{\prime} \leq 0,\\
& A^{\prime}-\id \leq 0,\\
& \tr{A^{\prime}(\vec{n}_1\cdot\vec{\Sigma})} = 0,\\
& \tr{A^{\prime}(\vec{n}_2\cdot\vec{\Sigma})} = 0,
\end{aligned}
\end{align}
The vectors $\vec{n}_{1,2}$ are orthogonal to $\vec{n}$ and linearly independent. The last two constraints ensure that the resulting effect $A$ is indeed in the direction defined by $\vec{n}$.

The solution given in \eqref{eq:suffA} is a feasible solution for this primal SDP problem and it thus means that $s_{\textrm{P}} \geq s_{max}$, where $s_{max}$ is given in \eqref{eq:smax}. 

\subsubsection*{Dual SDP problem}
The previous convex optimization problem can be transformed into a dual problem by standard methods \cite{CO04} (see Appendix):
\begin{align}
s_{\textrm{D}}:=&\min_{\lambda,\vec m} \tr{\lambda}\label{eq:dual}\\
&
\begin{aligned}\notag
\text{subject to}\qquad & \lambda\geq 0,\\
& \lambda\geq \vec m\cdot\vec \Sigma,\\
& \vec m\cdot\vec n=1.
\end{aligned}
\end{align}
It is always true that $s_{\textrm{D}} \geq s_{\textrm{P}}$.
Since we have already shown that $s_{\textrm{P}} \geq s_{max}$, the remaining thing is to show that $s_{max} \geq s_{\textrm{D}}$.
We do this by providing a corresponding feasible solution for the dual SDP problem. 

We denote by $Q$ the diagonal matrix with entries $Q_{ij}=\delta_{ij}p_i[+]$.
With this notation the vector $\vec{n}'$ in \eqref{eq:nprime} can be concisely written as $\vec{n}'=Q^{-1}\vec{n}/\no{Q^{-1}\vec{n}}$. 
We then set 
\begin{equation}
\label{eq:dualvars}
\vec{m}=\frac{Q^{-2}\vec n}{\|Q^{-1}\vec n\|^2} \, , \qquad 
\lambda=A^{\prime}(\vec m\cdot \vec \Sigma)A^{\prime} \, , 
\end{equation}
where $A^{\prime}=\A^{\prime}(+)$ is given in \eqref{eq:suffA}.
This choice of $\vec{m}$ fulfills the third dual condition and moreover we have
\begin{align}
\tr{\lambda}&=\tr{A'(\vec m\cdot\vec \Sigma)A'}=\tr{A'(\vec m\cdot\vec \Sigma)}\notag\\
&=\tr{A'(\vec n\cdot\vec \Sigma)}= s_{max} \, ,
\end{align}
where we first used the definition of $\lambda$, then in the next equality we used that $(A')^2=A'$ and in the third equality we used the fact that the trace for components of $\vec m$ orthogonal to $\vec n$ is zero.

It remains to show that $\lambda\geq 0$ and $\lambda\geq \vec{m}\cdot \vec{\Sigma}$. For this we will rewrite $\lambda$ in the basis of eigenvectors of $A^{\prime}$. 
The operator $A^{\prime}$ is a two-dimensional projection, hence it has doubly degenerate eigenvalues 1 and 0.
The eigenvalue 1 eigenvectors can be chosen to be
\begin{equation}
v_\pm=\frac{1}{2\sqrt{1\pm n_1'}}
\begin{pmatrix}
n_2'\pm in_3'\\
\mp n_2'-in_3'\\
1\pm n_1'\\
\pm i(1\pm n_1')
\end{pmatrix}.
\end{equation}
The eigenvalue 0 eigenvectors can be chosen to be
\begin{equation}
u_\pm=\frac{1}{2\sqrt{1\pm n_1'}}
\begin{pmatrix}
-n_2'\pm in_3'\\
\mp n_2'+in_3'\\
1\pm n_1'\\
\mp i(1\pm n_1')
\end{pmatrix}.
\end{equation}
Note that these eigenvectors are valid when $n_1'\neq\pm 1$; the excluded points have a different eigenbasis, however the remainder of the solution is analogous to what follows.
The eigenbasis $(v_+,v_-,u_+,u_-)$ is orthonormal and the operator $\vec m\cdot\vec \Sigma$ in this basis has a block-diagonal form,
\begin{equation}
\vec m\cdot\vec\Sigma=\begin{pmatrix}
M & 0\\
0 & -M^*
\end{pmatrix},
\end{equation}
where
\begin{equation}
M=\frac{s}{2}(\id+\vec g\cdot \vec \sigma)
\end{equation}
is a qubit operator with
\begin{equation}
\vec g=\frac{1}{\sqrt{1-n_1'^2}}
\begin{pmatrix}
\frac{p_-[2]}{p_+[2]}n_2'^2-\frac{p_-[3]}{p_+[3]}n_3'^2\\
\left(\frac{p_-[2]}{p_+[2]}+\frac{p_-[3]}{p_+[3]}\right)n_2'n_3'\\
-n_1'\sqrt{1-n_1'^2}\frac{p_-[1]}{p_+[1]}
\end{pmatrix}.
\end{equation}
Hence we have
\begin{equation}
\lambda=\begin{pmatrix}
M & 0\\
0 & 0
\end{pmatrix}, \qquad
\lambda -\vec m\cdot\vec\Sigma=\begin{pmatrix}
0 & 0\\
0 & M^{\ast}
\end{pmatrix}. 
\end{equation}
Therefore we need to check only the positivity of $M$.
The positivity of $M$ is in this case equivalent to the condition $\|\vec g\|\leq 1$. Using the fact that $(p_-[j]/p_+[j])^2\leq 1$ this is easily checked. 

To sum up, we have shown that our choices for primary and dual variables are feasible solutions that lead to the same values. Hence, this choice is optimal and the boundary given by $s_{max}$ given in (\ref{eq:smax}) is not only sufficient but also necessary.

\section{Examples}\label{sec:examples}

We will now demonstrate the use of the presented compatibility condition by looking at some concrete classes of qubit channels. 

\subsection{Partially depolarizing channels}

A partially depolarizing channel is an example of a Pauli channel. 
It is constructed as a mixture of the identity channel and the completely depolarizing channel $\dep$ to the maximally mixed state $\half\id$. 
As the mixing weight can vary, we get a one-parameter class of channels
\begin{equation}\label{eq:p}
\dep_p(\varrho)=(1-4p)\varrho+2p\id \, ,
\end{equation}
where $p\in [0,1/4]$.
The channel $\dep_p$ is, in fact, a Pauli channel with the probability vector $(1-3p,p,p,p)$, see also Fig.~\ref{fig:triangle}.
Actually, the map $\dep_p$ defined in \eqref{eq:p} is a valid channel for any $p\in[0,1/3]$, although the interpretation as a partially depolarizing channel holds only for $p\in [0,1/4]$.
Further, we can start from any unitary Pauli channel instead of the identity channel, however, in that case the depolarization is with respect to a different basis.

For a channel $\dep_p$ we have 
\begin{equation}
{p}_+[j]=2\left(p+\sqrt{p(1-3p)}\right)
\end{equation}
for every $j=1,2,3$. 
From Thm. \ref{prop:maincomp} we conclude that the set of unbiased qubit observables compatible with $\dep_p$ corresponds to the shrunken Bloch ball with the radius $2\left(p+\sqrt{p(1-3p)}\right)$.
The identity channel $id=\dep_0$ shrinks the compatibility Bloch ball region to the central point, while the completely depolarizing channel $\dep=\dep_{1/4}$ keeps the Bloch ball invariant.
This is consistent with our earlier observations.

An interesting special case is the universal quantum NOT channel, which transforms any qubit input state to as close as possible to its orthogonal complement \cite{BuHiWe99}. This operation cannot be perfect for any input and for general input is described by a Pauli channel falling under the case presented in this subsection, where $p=1/3$; in this case $s\leq 2/3$.

\subsection{Phase damping channels and L\"uder's channels}

A \emph{phase damping channel} is a map that damps the off-diagonal elements of a density matrix in a specific basis. 
We fix the basis to be the eigenbasis of $\sigma_3$. 
The action of a phase damping channel $\Phi_p$ is then such that in the $\sigma_3$-eigenbasis the density matrices retain their diagonals, but the off-diagonal elements acquire a factor of $2p-1$.
Thus, we have a one-parameter class of Pauli channels $\Phi_p$ and the corresponding probability vector is $(p,0,0,1-p)$.
For $p\in[1/2,1]$ the action of the channel describes pure damping, while for $p\in[0,1/2]$ the damping is complemented with the inversion of the off-diagonal elements.
The extreme case $p=0$ corresponds to the inversion in the $xy$-plane without any damping.

For a phase damping channel $\Phi_p$ we have ${p}_+[1]={p}_+[2]=0$ and ${p}_+[3]=2\sqrt{p(1-p)}$. 
Therefore, using Theorem \ref{prop:maincomp}, we conclude that an observable $\A_{s,\vec{n}}$ is compatible with $\Phi_p$ if an only if $\vec n = (0,0,1)$ and $s\leq 2\sqrt{p(1-p)}$.

Specific cases are the identity channel ($p=1$) and the NOT channel ($p=0$) for which $s=0$. On the other end lies the case of $p=1/2$ which is the completely phase damping channel that zeroes all off-diagonal elements and conserves the diagonal which contains all the information about $z$-direction; this means that all $z$-measurements are compatible with this channel ($s\leq 1$) --- see Fig.~\ref{fig:ellipsoid} on the right for this example.

An interesting class of channels falling into this category are L\"uder's channels of $\Z_t$, given as
\begin{equation}
\luders{\Z_t}(\varrho)= \sqrt{\Z_t(+)} \varrho \sqrt{\Z_t(+)} + \sqrt{\Z_t(-)} \varrho \sqrt{\Z_t(-)} \, .
\end{equation}
This is a phase damping channel $\Phi_p$ in the $\sigma_3$-eigenbasis with $p=\half(\sqrt{1-t^2}+1)$. 
One direct consequence is hence that an observable $\X_{s}$ is compatible with the L\"uder's channel $\luders{\Z_t}$ if and only if $s=0$. 
This result stands in contrast to the compatibility at the level of observables, as $\X_s$ and $\Z_t$ are compatible if and only if $s^2+t^2\leq 1$ \cite{Busch86}. 
This means that if we want to implement a joint measurement of $\X_s$ and $\Z_t$ with $s,t\neq 0$, the joint measurement process cannot contain a L\"uders channel $\luders{\Z_t}$ or $\luders{\X_s}$.
We remark that it has been earlier shown that a joint measurement process cannot contain both $\luders{\Z_t}$ and $\luders{\X_s}$ \cite{HeJiReZi10}.
Our new result hence strengthens that observation.

\subsection{Measure-and-prepare channels}

Let us consider the class of measure-and-prepare channels related to observables $\Z_t$, $0\leq t\leq 1$.
For each $0\leq t \leq 1$, we define a map $\Theta_t$ as
\begin{equation}
\Theta_t(\varrho) = \tr{\varrho \Z_t(+)} \Z_1(+) + \tr{\varrho \Z_t(-)} \Z_1(-) \, .
\end{equation}
This is a measure-and-prepare channel, which can be implemented by first measuring the observable $\Z_t$ and then preparing either pure state $\Z_1(+)$ or $\Z_1(-)$, depending on the outcome of the measurement.
It is easy to see that $\Theta_t$ is unital for any $t$. 
Also, this channel is a composition of a partially depolarizing channel and a completely phase damping channel in $z$-direction. 
Overall, $\Theta_t$ is a Pauli channel defined by the probability vector
\begin{equation}
\vec p=\frac{1}{4}(1+t,1-t,1-t,1+t) \, .
\end{equation}
See also Fig.~\ref{fig:triangle}.

When considering the compatibility of $\Theta_t$ with an observable $\A_{s,\vec{n}}$, Theorem \ref{prop:maincomp} gives
\begin{equation}
\frac{n_1^2+n_2^2}{1-t^2}+n_3^2\leq \frac{1}{s^2} \, .
\end{equation}
After some manipulation we get
\begin{equation}
s^2+t^2-s^2t^2\cos^2\vartheta \leq 1 \, , 
\end{equation}
where $\vartheta$ is the angle between the Bloch vector $\vec{n}$ and the $z$-axis. 
This condition is equivalent to the compatibility between the observables given in \cite{Busch86}. 
Specifically for observables $\X_s$ and $\Z_t$ the previous condition on compatibility gives
\begin{equation}
s^2+t^2\leq 1 \, . 
\end{equation}
This is what one would have expected due to the physical nature of $\Theta_t$.

\section{Summary}
We addressed the question of compatibility of unbiased qubit channels and observables. 
Although our analysis was made explictly for Pauli channels, i.e.~random mixtures of Pauli unitaries, the conclusions hold for a more general case as well --- without lost of generality \cite{RuSzWe02} any qubit unital channel $\Phi$ can be expressed as a convex combination of (at most) four orthogonal unitary channels induced by unitary operators $U_j=U\sigma_j V^*$, where $U,V$ are suitable unitary operators. It follows that $\tr{U_j^* U_k}=0$ for $j\neq k$ and $\Phi(\varrho)=U\Psi_{\vec{p}}[V^*\varrho V]U^*$.

We have derived a compatibility formula (Eq.~\eqref{eq:tuomas}) for the case of unbiased qubit channels and observables. We have shown (Theorem \ref{prop:maincomp}) that for a given unital qubit channel the set of compatible unbiased observables form an ellipsoid (see Fig.~\ref{fig:ellipsoid}) naturally embedded inside the Bloch sphere representing the set of binary unbiased observables. 
Let us stress that concerning the observable compatibility by Proposition \ref{prop:unitary} the rotation induced by $U$ is irrelevant. Therefore, the ellipsoid for $\Phi$ is connected to the ellipsoid for the Pauli channels $\Psi_{\vec{p}}$ by the unitary rotation $V$. For a given qubit unital channel the sharpest (least noisy) compatible observable (quantified by parameter $s$) is oriented along the direction for which the value of $p_+[j]$ is maximal. 

Alternatively, we can depict also a region of Pauli channels compatible with an observable, however, this view is not so enlightening (see Fig.~\ref{fig:uglyregion}).

\section*{ACKNOWLEDGEMENTS}

T.H. acknowledges financial support from the Horizon 2020 EU collaborative
projects QuProCS (Grant Agreement No. 641277) and the Academy of Finland (Project no. 287750).
D.R. was financed by SASPRO Program No. 0055/01/01 project QWIN cofunded by the European Union and Slovak Academy of Sciences. We also acknowledge support from VEGA 2/0151/15 project QWIN and APVV-14-0878 project QETWORK. 
We are also grateful to David Reeb for many fruitful scientific discussions and for helping us to learn SDP before his scientific retirement.

\section*{APPENDIX: THE DUAL SDP PROBLEM}

Let us recall the primal problem
\begin{align}
s_{\textrm{P}}:=&\max_{A'} \tr{\A^{\prime}(\vec{n}\cdot\vec{\Sigma})}\\
&
\begin{aligned}\notag
\text{subject to}\qquad & -\A^{\prime} \leq 0,\\
& \A^{\prime}-\id \leq 0,\\
& \tr{\A^{\prime}(\vec{n_1}\cdot\vec{\Sigma})} = 0,\\
& \tr{\A^{\prime}(\vec{n_2}\cdot\vec{\Sigma})} = 0,
\end{aligned}.
\end{align}
The Lagrangian is then
\begin{align}
&L(\A^{\prime},\lambda_1,\lambda_2,\mu,\nu)=\nonumber\\
&\tr{\A^{\prime}(-(\vec{n}+\mu \vec{n}_1+\nu \vec{n}_2)\cdot\vec{\Sigma}-\lambda_1+\lambda_2)}-\tr{\lambda_2},
\end{align}
where $\lambda_i\geq 0$.
The Lagrange dual function $g(\lambda_1,\lambda_2,\mu,\nu)$ is then infimum over $\A^{\prime}$ of the Lagrangian. Since it is linear in $\A^{\prime}$ we get that
\begin{align}
 g(\lambda_1,\lambda_2,\mu,\nu)= \left\{
\begin{array}{ll}
      -\tr{\lambda_2} & -(\vec{n}+\mu \vec{n}_1+\nu \vec{n}_2)\cdot\vec{\Sigma}\\
     &\qquad\quad -\lambda_1+\lambda_2=0 \\
      -\infty & \textrm{otherwise}. \\
\end{array} 
\right. 
\end{align}
Thus the function $g$ is nontrivial only when 
\begin{equation}
\lambda_2-(\vec{n}+\mu \vec{n}_1+\nu \vec{n}_2)\cdot\vec{\Sigma}=\lambda_1\geq 0 \, .
\end{equation}
 Let $\vec{n}+\mu \vec{n}_1+\nu \vec{n}_2=\vec{m}$. Note that $\vec{m}\cdot\vec{n}=1$. We can include this condition into the constraints of the dual problem, which is now stated as
\begin{align}
s_{\textrm{D}}:=&\max_{\lambda_1,\lambda_2,\vec m} (-\tr{\lambda_2})\\
&
\begin{aligned}\notag
\text{subject to}\qquad & \lambda_2\geq 0,\\
& \lambda_2 - \vec m\cdot\vec \Sigma=\lambda_1,\\
& \lambda_1\geq 0\\
& \vec m\cdot\vec n=1,
\end{aligned}
\end{align}
or in its simplified form in \eqref{eq:dual}.

\vspace{0.5cm}

\end{document}